\documentclass[preprint]{elsarticle}
\usepackage[latin1]{inputenc}
\usepackage[danish, english]{babel}
\usepackage[T1]{fontenc}
\usepackage{amsmath,amssymb}
\usepackage{algorithmicx}
\usepackage{algpseudocode}
\algrenewcommand\algorithmicindent{1.0em}%
\usepackage{comment}
\usepackage{graphicx}
\usepackage{calc}
\usepackage{tabularx}
\usepackage{url}
\usepackage[capitalise]{cleveref}
\usepackage{textcomp}

\newtheorem{definition}{Definition}
\newtheorem{example}{Example}

\newtheorem{proof}{Proof}

\newtheorem{lemma}{Lemma}
\newtheorem{proposition}{Proposition}

\newcommand{\mc}{\mathcal}

\newcommand{\rhs}{{\sf rhs}}
\newcommand{\lhs}{{\sf lhs}}
\newcommand{\flhs}{{\sf lhsf}}
\newcommand{\Lhs}{{\sf Lhsf}}
\newcommand{\Rhs}{{\sf Rhs}}

\newcommand{\func}{{\sf func}}
\newcommand{\fmap}{{\cal T}}
\newcommand{\arity}{{\sf ar}}
\newcommand{\any}{\textit{any}}

\newcommand{\Q}{{\cal Q}}

\newcommand{\num}{\textit{num}}

\newcommand{\listlist}{\textit{listlist}}
\newcommand{\listtype}{\textit{list}}
\newcommand{\e}{\textit{e}}

\newcommand{\Term}{{\sf Term}}

\newcommand{\Lang}{L}

     {\end{tabular}\end{tt}\end{small}}
\def\anno#1{{\ooalign{\hfil\raise.07ex\hbox{\small{\rm #1}}\hfil%
        \crcr\mathhexbox20D}}}

\journal{JLAMP}
\begin{document}
\begin{frontmatter}
\title{Optimised determinisation and completion of finite tree automata}
\author[roskilde,imdea]{John~P.~Gallagher\corref{mycorrespondingauthor}}
\cortext[mycorrespondingauthor]{Corresponding author}
\ead{jpg@ruc.dk}
\author[itu]{Mai~Ajspur}
\author[melb]{Bishoksan~Kafle}

\address[roskilde]{Roskilde University, Denmark}
\address[imdea]{IMDEA Software Institute, Madrid, Spain}
\address[itu]{IT University of Copenhagen, Denmark}
\address[melb]{The University of Melbourne, Australia}

\begin{abstract}
Determinisation and completion of finite tree automata are important operations with applications in program analysis and verification.  However, the complexity of the classical procedures for determinisation and completion is high. They are not practical procedures for manipulating tree automata beyond very small ones.  In this paper we develop an algorithm for determinisation and completion of finite tree automata, whose worst-case complexity remains unchanged, but which performs far better than existing algorithms in practice. The critical aspect of the algorithm is that the transitions of the determinised (and possibly completed) automaton are generated in a potentially very compact form called product form, which can reduce the size of the representation dramatically. Furthermore, the representation can often be used directly when manipulating the determinised automaton. The paper contains an experimental evaluation of the algorithm on a large set of tree automata examples. 
\end{abstract}
\end{frontmatter}

%
\section{Introduction}\label{intro}

A recognisable tree language is a possibly infinite set of trees that are accepted by a finite tree automaton (FTA).   FTAs and the corresponding recognisable languages have desirable properties such as closure under Boolean set operations, and decidability of membership and emptiness.  

In the paper we will give a brief overview of the relevant features of FTAs, but the main goal of the paper is to focus on two operations on FTAs, namely \emph{determinisation} and \emph{completion}.  These operations play a key role in the theory of FTAs, for example in showing that recognisable tree languages are closed under Boolean operations.  Potentially, they also play a practical role in systems that manipulate sets of terms, but their complexity has so far discouraged their widespread application.

In the paper we develop an optimised algorithm that performs determinisation and optionally completion, and analyse its properties. The most critical aspect of the optimisation is a compact representation of the set of transitions of the determinised automaton, called \emph{product transitions}. Experiments show that the algorithm performs well, though the worst case remains unchanged.  We also discuss applications of finite tree automata that exploit the determinisation algorithm. 

In Section \ref{prelim} the essentials, for our purposes, of finite tree automata are introduced, including the notion of product transitions.  The operations of determinisation and completion are defined.  Section \ref{algorithm} presents the optimised algorithm for determinising an FTA.  It is developed in a series of stages starting from the textbook algorithm for determinisation. 
In Section \ref{algorithm} it is shown how the algorithm can be optimised and output transitions in product form. The performance of the algorithm is analysed in Section \ref{complexity}. In Section \ref{completion} we discuss the combination of determinisation and completion of an FTA and show that the performance of the algorithm generating product form is as effective when generating a complete determinised automaton. Section \ref{experiments} reports on the performance of the algorithm on a large number of example tree automata. Section \ref{applications} discusses potential applications of the algorithm to problems in program analysis and verification, and also to  tree automata problems including checking inclusion and universality., Section \ref{related} contains a discussion of related work and finally in Section \ref{future} we summarise and discuss further work and applications.

 %
\section{Preliminaries}\label{prelim}

A {\em finite tree automaton}
(FTA) is a quadruple 
$\langle Q, Q_f, \Sigma, \Delta \rangle$,
where 
\begin{enumerate}
\item
$Q$ is a finite set called {\em states}, 
\item
$Q_f \subseteq Q$ is called the set of accepting (or final) states,
\item
$\Sigma$ is a set of  
function symbols (called the \emph{signature}) and 
\item
$\Delta$ is a set of 
{\em transitions}. 
\end{enumerate}
Each function symbol $f \in \Sigma$ has an arity $n \ge 0$, written $\arity(f)=n$.
Function symbols with arity 0 are called {\em constants}. $Q$ and $\Sigma$ are disjoint. 
$\Term({\Sigma})$ is the 
set of {\em ground terms} (also called {\em trees}) constructed from $\Sigma$ where $t \in \Term({\Sigma})$ iff $t \in \Sigma$ is a constant or $t = f(t_1,\ldots,t_n)$ where $\arity(f)=n$ and $t_1,\ldots,t_n \in  \Term({\Sigma})$.
Similarly $\Term(\Sigma \cup Q)$ is the set of terms/trees constructed from $\Sigma$ and $Q$, treating the elements of $Q$ as constants.
Each transition in $\Delta$ is of the form
$f(q_1,\ldots,q_n) \rightarrow q$, where
$\arity(f)=n$ and $q,q_1,\ldots,q_n \in Q$.

To define acceptance of a term by the FTA $\langle Q, Q_f, \Sigma, \Delta \rangle$ we first define a {\em context} for the FTA. A context is a term from $\Term(\Sigma \cup Q \cup \{ \bullet\})$ containing exactly one occurrence of $\bullet$ (which is a constant not in $\Sigma$ or $Q$).  Let $c$ be a context and $t  \in \Term(\Sigma \cup Q)$;  $c[t]$ denotes the term resulting from the replacement of $\bullet$ in $c$ by $t$.  A term $t \in \Term(\Sigma \cup Q)$ can be written as $c[t']$ if $t$ has a subterm $t'$, where $c$ is the context resulting from replacing that subtree by $\bullet$.

The binary relation $\Rightarrow$ represents one step of a (bottom-up) run for the FTA.  It is defined as follows; $c[l] \Rightarrow c[r]$ iff $c$ is a context and $l \rightarrow r \in \Delta$.   The reflexive, transitive closure of $\Rightarrow$ is denoted $\Rightarrow^*$.

A run for $t \in \Term(\Sigma)$ exists if $t \Rightarrow^* q$ where $q \in Q$.  The run is \emph{successful} if $q \in Q_f$ and in this case $t$ is {\em accepted} by the FTA.  There may be more than one state $q$ such that $t \Rightarrow^* q$  and hence FTAs are sometimes called NFTAs, where N stands for nondeterministic.
A tree automaton $R$ defines a set of terms, that is, a tree language,
denoted $L(R)$, as the set of all terms that it accepts. We also write $L(q)$ to be the set of terms $t$ such that $t \Rightarrow^* q$ in a given FTA.

\begin{definition}\label{bu-det}
An $FTA$ $\langle Q, Q_f, \Sigma, \Delta \rangle$ is called {\em bottom-up deterministic} if and only if $\Delta$ contains no two transitions with the same left hand side.  A bottom-up deterministic FTA is abbreviated as a DFTA.
\end{definition}
Runs of a DFTA are deterministic in the following sense;  for every context $c$ and term of form $c[t]$ there is at most one term $c[t']$ such that $c[t] \Rightarrow c[t']$.  It follows that for every $t \in \Term(\Sigma)$ there is at most one $q \in Q$ such that $t \Rightarrow^* q$. 
As far as expressiveness is concerned
we can limit our attention to DFTAs\footnote{We do not deal here with top-down deterministic FTA, which are strictly less expressive than FTAs.}.  For every
FTA $R$ there exists a DFTA $R'$ such that
$L(R) = L(R')$.

\begin{definition}\label{complete}
An automaton $R = \langle Q, Q_f, \Sigma, \Delta \rangle$ is {\em complete} if 
for all n-ary functions $f \in \Sigma$
and states $q_1,\ldots,q_n \in Q$, there exists a state $q \in Q$ such that
$f(q_1,\ldots,q_n) \rightarrow q  \in \Delta$.
\end{definition}
It follows that in a complete FTA every term $t$ has at least one run and furthermore in a complete DFTA each $t$ has a run to exactly one state. Thus a complete DFTA defines a partition of $\Term(\Sigma)$, namely $\{L(q) \mid q \in Q\} \setminus \{\emptyset\}$.

\begin{figure}
\begin{center}
\includegraphics[width=4.3 in]{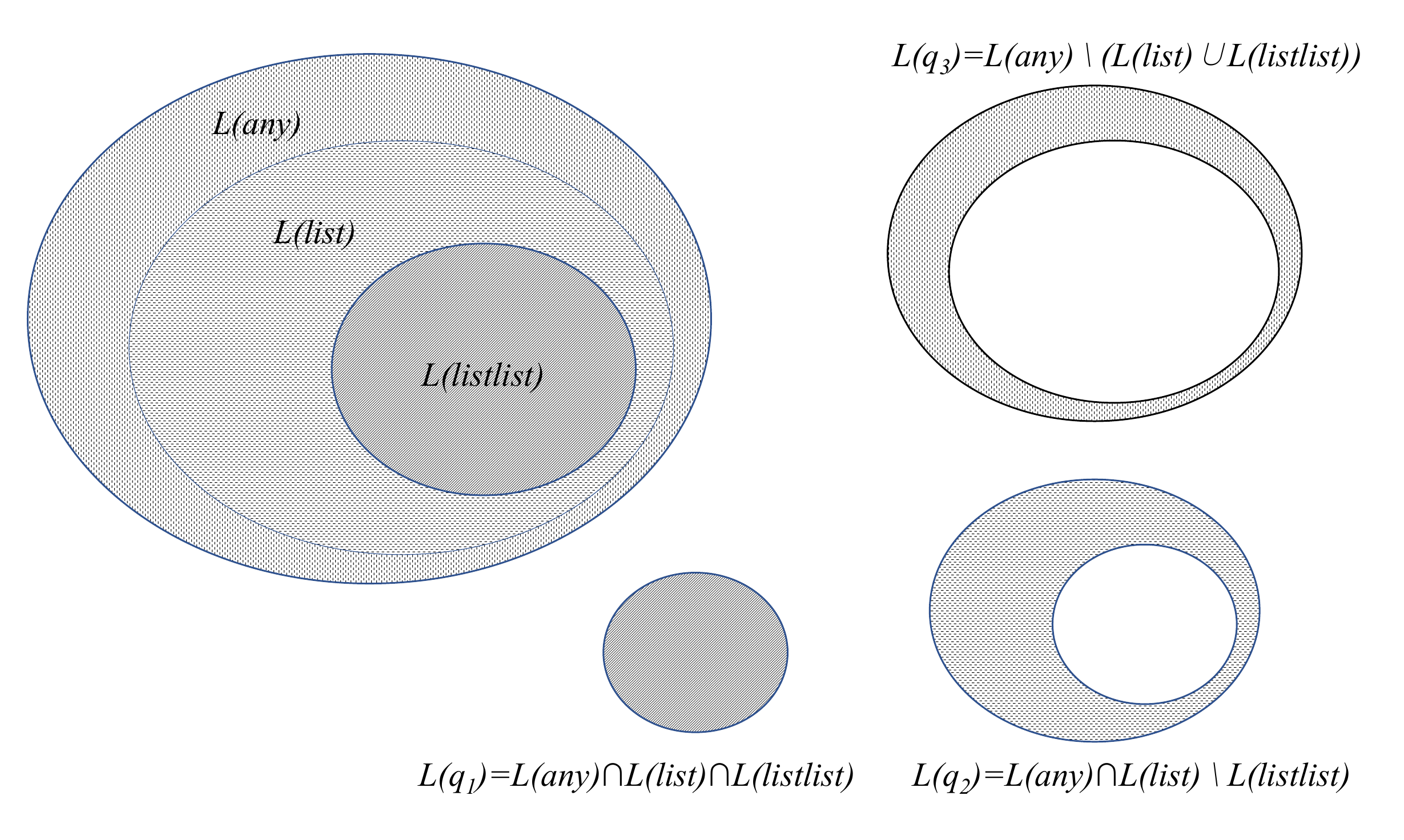}
\end{center}
\caption{The disjoint languages from Example \ref{ex2}}
\protect\label{fig-ex1}
\end{figure}

\begin{definition}\label{delta_any}
Let $\Sigma$ be a signature and ``$\any$'' a state. We define $\Delta_{\any}^{\Sigma}$ to be the following
set of transitions.
$$\{f(\stackrel{n {\rm ~times}}{\overbrace{\any,\ldots,\any}}) \rightarrow \any ~|
f^n \in \Sigma\}$$
Clearly, given an FTA $\langle Q, Q_f, \Sigma, \Delta \rangle$ with $\any\in Q$ and $\Delta_{\any}^{\Sigma} \subseteq \Delta$, there is a run $t \Rightarrow^* \any$ for any $t\in \Term(\Sigma)$, that is, $L(\any) = \Term(\Sigma)$.
\end{definition}
\noindent
We normally drop the superscript in $\Delta_{\any}^{\Sigma}$ as $\Sigma$ is usually clear from the context.
\begin{example}\label{ex2}

Let $\Sigma = \{[], [.|.], 0\}$, 
$Q = \{\listtype,\listlist,\any\}$, $Q_f = \{\listtype, \listlist\}$ and $\Delta = 
\{[] \rightarrow \listtype, [\any | \listtype]  \rightarrow \listtype,
[] \rightarrow \listlist, [\listtype | \listlist]  \rightarrow \listlist \} \cup  \Delta_{\any}$.  $L(\listtype)$ is
 the set of lists of
any terms, while $L(\listlist)$ is the set of lists whose elements are
themselves lists.  Clearly $L(\listlist)$ is contained in $L(\listtype)$, which is contained in $L(\any)$.

The automaton is not bottom-up deterministic since $[]$ occurs as the left hand side in more than one transition; a determinisation algorithm (see Section \ref{algorithm}) yields the DFTA
$\langle Q', Q'_f, \Sigma, \Delta' \rangle$, where $Q' = \{ q_1,q_2,q_3\}$,
$Q'_f = \{q_1,q_2\}$ and
 $\Delta' = \{[] \rightarrow q_1,
[q_1|q_1]  \rightarrow q_1, [q_2|q_1]  \rightarrow q_1,
[q_1|q_2]  \rightarrow q_2, [q_2|q_2]  \rightarrow q_2,
[q_3|q_2]  \rightarrow q_2, [q_3|q_1]  \rightarrow q_2, 
[q_2|q_3]  \rightarrow q_3, [q_1|q_3]  \rightarrow q_3, 
[q_3|q_3]  \rightarrow q_3, 0 \rightarrow q_3\}$.

\noindent
The states $q_1$, $q_2$ and $q_3$ are abbreviations for elements of the powerset of the states of the original FTA; here $q_1 = \{\any,\listtype,\listlist\}$, $q_2 =\{\any,\listtype\}$ and $q_3 = \{\any\}$. This automaton is also complete.  
\qed
\end{example}
In Example \ref{ex2}, $L(q_1)=
L(\any) \cap L(\listtype) \cap L(\listlist)$, $L(q_2)=(L(\listtype) \cap L(\any)) \setminus L(\listlist)$, and $L(q_3)=
L(\any) \setminus (L(\listtype) \cup L(\listlist))$.    
The relationship between the languages corresponding to the FTA and DFTA states in Example \ref{ex2} is shown in
Figure \ref{fig-ex1}. 

A critical aspect of this paper is a compact representation for the set of transitions $\Delta$, using the notion of product transition defined as follows.

\begin{definition}[Product transition]\label{product-trans}
Let $\langle Q,F,\Sigma,\Delta\rangle$ be an FTA. A {\em product transition} is of the form
$f(Q_1,\ldots,Q_n) \rightarrow q$ where $Q_i \subseteq Q, 1 \le i \le n$  
and $q \in Q$.  This product transition denotes the set of transitions
$\{f(q_1,\ldots,q_n) \rightarrow q ~|~q_1 \in Q_1,\ldots,q_n \in Q_n\}$.
Thus $\prod^{n}_{i=1} \vert Q_i\vert$  transitions are represented
by a single product transition. 
\end{definition}
Instead of expanding the product transitions, we can regard a product transition as introducing $\epsilon$-transitions.  An $\epsilon$-transition has the form $q_1 \rightarrow q_2$ where $q_1,q_2$ are states.  $\epsilon$-transitions can be eliminated, if desired. Given a product transition $f(Q_1, \ldots, Q_n) \rightarrow q$, introduce $n$ new non-final states $s_1,\ldots,s_n$ corresponding to $Q_1, \ldots, Q_n$ respectively and replace the product transition by the set of transitions $\{f(s_1,\ldots,s_n) \rightarrow q\} \cup \{ q' \rightarrow s_i \mid q' \in Q_i, 1 \le i \le n\}$.  It can be shown that this transformation preserves the language of the FTA. 

\begin{example}\label{ex3}
The transitions of the DFTA generated in Example \ref{ex2} can be represented in product
transition form as follows.
\[
\begin{array}{lll}
 \Delta' =& \{[] \rightarrow q_1 & 0 \rightarrow q_1\\
&[\{q_1,q_2\}|\{q_1\}]  \rightarrow q_1~~~~~~
&[\{q_3\}|\{q_1\}]  \rightarrow q_2\\
&[\{q_1,q_2\}|\{q_3\}]  \rightarrow q_3
&[\{q_3\}|\{q_3\}]  \rightarrow q_3\\
&[\{q_1,q_2\}|\{q_2\}]  \rightarrow q_2
&[\{q_3\}|\{q_2\}]  \rightarrow q_2\}\\
\end{array}
\]
These 8 product transitions represent the 11 transitions shown in 
Example \ref{ex2}.
There are more compact equivalent sets of product transitions, for example.
\[
\begin{array}{lll}
 \Delta'' =& \{[] \rightarrow q_1 & 0 \rightarrow q_1\\
&[\{q_1,q_2,q_3\}|\{q_3\}]  \rightarrow q_3~~~~~~~
&[\{q_1,q_2,q_3\}|\{q_2\}]  \rightarrow q_2\\
&[\{q_1,q_2\}|\{q_1\}]  \rightarrow q_1
&[\{q_3\}|\{q_1\}]  \rightarrow q_2\}\\
\end{array}
\]
 \qed
\end{example}

In the determinisation algorithm to be developed in the next section, product transitions for the output DFTA are generated directly; this leads in many cases to an exponential saving of space and time, as shown by the experiments in Section \ref{experiments}.

 %

\section{Development of an Optimised Determinisation Algorithm}\label{algorithm}
\begin{figure}[t]
  \centering
 \fbox{
\begin{minipage}{\textwidth}
  \begin{algorithmic}[1]
 \Procedure{FTA Determinisation }{Input: $\langle Q,Q_f,\Sigma,\Delta\rangle$}
\State $\Q_d\gets \emptyset$
\State $\Delta_d\gets \emptyset$
\Repeat
\State $\Q_d\gets \Q_d \cup \{Q_0\}$,
	\State $\Delta_d \gets \Delta_d\cup\{f(Q_1,\ldots,Q_{\arity(f)})\to Q_0\}$
	\State \textbf{where}
		\State \indent $f^n\in\Sigma, \ \ Q_1,\ldots,Q_{\arity(f)}\in \Q_d$,
		\State \indent $Q_0 = \{q_0\;|\; \exists q_1\in Q_1,\ldots, q_{\arity(f)}\in Q_{\arity(f)}, (f(q_1,\ldots, q_{\arity(f)})\to q_0)\in\Delta\}$ 
\Until{ no rules can be added to $\Delta_d$}
\State $\Q_{f} \gets \{Q' \in \Q_d \;|\; Q' \cap Q_f \neq\emptyset\}$
\State \Return $(\Q_d,  \Q_{f}, \Sigma, \Delta_d)$
\EndProcedure
\end{algorithmic}
\end{minipage}
}
  \caption{Textbook Determinisation Algorithm}
  \label{fig-alg1}
\end{figure}

In this section we present the textbook algorithm for FTA determinisation, and then proceed to optimise it.  
The starting point is the determinisation algorithm in Figure \ref{fig-alg1}, which is taken (apart from some renaming of variables) from \cite{Comon}.  

The main idea of the transformation is to delay the computation of transitions $\Delta_d$ until after the complete set of states of the DFTA, $\Q_d$, has been computed.  This has the following advantages.

\begin{itemize}
\item
It allows a suitable data structure to be built, from which the set of transitions can be generated in product form.
\item
It enables significant optimisation of the main \textbf{repeat} loop to avoid redundant computations.
\item
It allows the computation of the set of transitions to be omitted entirely if desired, thus permitting other applications of the algorithm such as inclusion checking, which only requires the states.
\end{itemize}

We note first a small ambiguity in the algorithm as presented in \cite{Comon}.  In the assignment $Q_0 = \{q_0\;|\; \exists q_1\in Q_1,\ldots, q_n\in Q_n, (f(q_1,\ldots, q_n)\to q_0)\in\Delta\}$ the right hand side is implicitly assumed to evaluate to a non-empty set, otherwise it is ignored.  Although allowing the variable $Q_0$ to take the value $\emptyset$ would return a correct result, many redundant transitions of the form $f(Q_1,\ldots,Q_n)\to \emptyset$ would be generated. In our transformed algorithm we make this assumption explicit and eliminate such transitions.

Note that the states of the computed DFTA are elements of $2^Q$ where $Q$ is the set of states of the input FTA.

\subsection{Introduction of functional notation}

Let $\langle Q,\Sigma,Q_f,\Delta \rangle$ be an FTA. Let $\delta = f(q_1,\ldots,q_n) \rightarrow q, n \ge 0$ be a transition in $\Delta$.  Define the following selector functions on $\delta$.
\[
\begin{array}{lll}
\rhs: \Delta \rightarrow Q ~~~~&\lhs_i: \Delta \hookrightarrow Q~~~&\func: \Delta \rightarrow Q\\
\rhs(\delta)  = q                              &\lhs_i(\delta) = q_i, 1 \le i \le n ~~~              &\func(\delta) = f

\end{array}
\]
The $\lhs_i$ functions are partial functions on $\Delta$ since $\lhs_i$ is not defined for every transition for a given $i$. In particular  the $\lhs_i$ functions are undefined on transitions whose function symbol has arity zero.

The inverse mappings $\lhs^{-1}_i:Q \rightarrow 2^\Delta$ and $\func^{-1}:\Sigma \rightarrow 2^\Delta$ are defined respectively as $\lhs^{-1}_i(q) = \{\delta \mid \lhs_i(\delta)=q\}$,   $\func^{-1}(f) = \{\delta \mid \func(\delta)=f\}$.   Using these, $\flhs_{i}:(\Sigma \times Q) \rightarrow 2^\Delta$ is defined as $\flhs_{i}(f,q) = \lhs^{-1}_i(q) \cap \func^{-1}(f)$.

$\flhs_{i}(f,q)$ can be regarded as an index for $\Delta$ returning the set of transitions whose function symbol is $f$ and whose left hand side has $q$ in the $i^{th}$ position. 
The mappings $\flhs_{i}$ are lifted to sets of states, giving $\Lhs_{i}$ defined as follows.
\[
\begin{array}{l}
\Lhs_{i}: (\Sigma \times 2^Q) \rightarrow 2^{\Delta}\\
\Lhs_{i}(f,S) = \bigcup_{s \in S} \flhs_{i}(f,s)\\

\end{array}
\]
We also lift $\rhs$ to sets of transitions, giving the function $\Rhs: 2^{\Delta} \rightarrow 2^Q$, where $\Rhs(T) = \{\rhs(\delta) \mid \delta \in T\}$.

The following property allows us to reformulate the algorithm, obtaining a form that is easier to manipulate.

\begin{lemma}\label{prop1}
The following expressions are equal for all $f \in \Sigma$ and $Q_1, \ldots Q_n \in 2^Q$.
\begin{equation}
\{q_0\;|\; \exists  q_1\in Q_1,\ldots, q_n\in Q_n, (f(q_1,\ldots, q_n)\to q_0)\in\Delta\}
\end{equation}
\begin{equation}
\begin{split}
{\rm if}~\arity(f)=0 &~{\rm then}~\Rhs(\func^{-1}(f))~{\rm else}~\\      &\Rhs(\Lhs_1(f,Q_1) \cap \cdots \cap \Lhs_{\arity(f)}(f, Q_{\arity(f)}))
\end{split}
\end{equation}

\end{lemma}
\begin{proof}
Application of the definitions of $\Rhs$, $\Lhs_i$ and set operations.
\qed
\end{proof}

\subsection{Delaying computation of transitions}\label{sec-trans}

Examining the \textbf{repeat} loop in Figure \ref{fig-alg1}, we observe that 
the values of $\Delta_d$ and $\Q_d$ are initialised to $\emptyset$ before the first iteration of the loop and recomputed on each iteration of the loop body. Let us  represent the action of the loop body as a function $\phi$, and express the effect of the loop body as $(\Q'_d, \Delta'_d)  \gets \phi( \Q_d,\Delta_d)$.

We can then summarise the execution of the \textbf{repeat} loop as the sequence 
\[
\begin{array}{l}
(\Q^1_d,\Delta^1_d) \gets \phi(\emptyset, \emptyset)\\
(\Q^2_d,\Delta^2_d) \gets \phi(\Q^1_d,\Delta^1_d)\\
 \ldots\\
 (\Q^{k-1}_d,\Delta^{k-1}_d) \gets \phi(\Q^{k-2}_d,\Delta^{k-2}_d)\\
 (\Q^k_d,\Delta^{k}_d) \gets \phi(\Q^{k-1}_d,\Delta^{k-1}_d)\\
\end{array}
\]
\noindent
and the loop terminates when $\Delta^k_d=\Delta^{k-1}_d$.  We can establish that the above iteration sequence could be replaced by the following one, in which the value of $\Delta_d$ does not accumulate but is computed ``from scratch'' on each iteration, in effect executing $\Delta_d \gets \emptyset$ at the start of each iteration.  This is because the set of transitions derived in each iteration depends only on the value of $\Q_d$ at the start of the loop body, and not on the current value of $\Delta_d$.
\[
\begin{array}{l}
(\Q^1_d,\Delta^1_d) \gets \phi(\emptyset, \emptyset)\\
(\Q^2_d,\Delta^2_d) \gets \phi(\Q^1_d,\emptyset)\\
 \ldots\\
 (\Q^{k-1}_d,\Delta^{k-1}_d) \gets \phi(\Q^{k-2}_d,\emptyset)\\
 (\Q^k_d,\Delta^{k}_d) \gets \phi(\Q^{k-1}_d,\emptyset)\\
\end{array}
\]
The sequence of values $(\Q^1_d,\Delta^1_d),\ldots,(\Q^k_d,\Delta^{k}_d)$ is the same in the two iteration sequences, except that possibly the second sequence is shorter by one, since the termination condition for the \textbf{repeat} loop is then changed to $\Q^k_d = \Q^{k-1}_d$, possibly reducing the number of iterations by one, as the value $\Q_d$ can stabilise one iteration earlier than the value of $\Delta_d$ does. Comparison of the sequences shows that the computation of $\Delta_d$ can be removed from the {\bf repeat} loop entirely, delaying it until after the termination of the loop. From the second sequence we can see that the final value of $\Delta_d$ can be computed from the final value of $\Q_d$ with one extra execution of the loop body. We return to the computation of $\Delta_d$ in Section \ref{transitions}.

The next stage of the transformed algorithm after applying these transformations is displayed in Figure \ref{fig-alg4}, where the processing of 0-arity functions, which depends only on the original FTA, is completed (lines 3-10) before entering the main loop.

The transformations so far are fairly superficial and have little bearing on the efficiency of the algorithm.  However they enable us to focus on the iterations of the inner {\bf for all} loop, with a view to more substantial  efficiency improvements.  

\begin{figure}[t]
  \centering

 \fbox{
\begin{minipage}{\textwidth}
\begin{algorithmic}[1]
\Procedure{FTA Determinisation }{Input: $\langle Q,\Sigma,Q_f,\Delta\rangle$}
\State $\Q_d\gets \emptyset$
\ForAll {$f \in \Sigma$}
    \If {$(\arity(f)=0)$}
	  \State  $Q_0 \gets \Rhs(\func^{-1}(f))$
	  \If {$Q_0 \neq \emptyset$}
		   \State $\Q_d \gets \Q_d \cup \{Q_0\}$
          \EndIf
     \EndIf
\EndFor
\Repeat
	\State $\Q_d^{old}\gets \Q_d$
	\ForAll {$f \in \Sigma$}
	\If {$(\arity(f)>0)$}
	\State $n \gets \arity(f)$
	\ForAll {$(Q_1, \ldots Q_n) \in (\Q_d^{old} \times \cdots \times \Q_d^{old})$}
		  \State $Q_0 \gets \Rhs(\Lhs_1(f,Q_1) \cap \cdots \cap \Lhs_n(f, Q_n))$ 
		\If {$Q_0 \neq \emptyset$}
		   \State $\Q_d \gets \Q_d \cup \{Q_0\}$
		\EndIf 
	\EndFor
	\EndIf
	\EndFor
\Until{$\Q_d = \Q_d^{old}$}
\State{Compute the set of transitions $\Delta_d$ (see Section \ref{transitions})}
\State $\mc{Q}_{f} \gets \{Q' \in \Q_d \;|\; Q' \cap Q_f \neq\emptyset\}$
\State \Return $(\Q_d, \Sigma, \mc{Q}_{f}, \Delta_d)$
\EndProcedure
\end{algorithmic}
\end{minipage}
}
 \caption{Algorithm after delaying the computation of transitions}
  \label{fig-alg4}
\end{figure}

\subsection{Inner Loop Optimisation}\label{step5}
The fact that we no longer need to compute transitions in the inner loop can lead to major savings since we can focus on the computation of $\Q_d$. 
Let us suppose that $|\Q_d^{old}| =k$ in Figure \ref{fig-alg4} (line 12). Then for a function symbol $f$ of arity $n$, there are  $k^n$  tuples $(Q_1, \ldots Q_n)$ in the cartesian product $(\Q_d^{old} \times \cdots \times \Q_d^{old})$ and so the function $\Lhs_i(f,Q_j)$ is called $n*k^n$ times. On the other hand, within the loop there are only $k*n$ different calls of the form $\Lhs_i(f,Q_j)$ and therefore it is worth precomputing these $k*n$ values outside the loop and avoid recomputing the same call many times.  Furthermore, cases of $\Lhs_i(f,Q_j)$ that evaluate to the empty set can be ignored since they cannot contribute to a non-empty value of $Q_0$ within the loop, since $\Rhs(\emptyset) = \emptyset$.

We precompute the $\Lhs_i(f,Q_j)$ values by introducing a function called $\fmap_i : (\Sigma \times 2^{2^Q}) \rightarrow 2^{2^{\Delta}}$ defined as 
\[
\fmap_i(f,\Q') = \{\Lhs_i(f,Q') \mid Q' \in \Q'\} \setminus \{\emptyset\}. 
\]
This function is defined for $1 \le i \le n$ for a function of arity $n$ and returns a set of sets of transitions. The inner
{\bf for all} loop is then rewritten to iterate over tuples of sets of transitions chosen from the product $\fmap_1(f,\Q_d^{old}) \times \cdots \times \fmap_n(f,\Q_d^{old})$ instead of $(\Q_d^{old} \times \cdots \times \Q_d^{old})$. 

It can be seen that the same non-empty values of $Q_0$ are generated within the loop. Whereas previously we iterated over tuples $(Q_1, \ldots Q_n)$ in the cartesian product $(\Q_d^{old} \times \cdots \times \Q_d^{old})$, and then applied  $\Rhs(\Lhs_1(f,Q_1) \cap \cdots \cap \Lhs_n(f, Q_n))$, now we precompute all the possible $\Lhs_i(f,Q_i)$ values (which are sets of transitions $\Delta_i$) and then call  $\Rhs(\Delta_1) \cap \cdots \cap \Delta_n)$ for every possible tuple of non-empty sets of transitions $(\Delta_1, \ldots \Delta_n)$. 

The transformation of the inner loop is significant in typical applications.  Instead of $k^n$ iterations of the loop, where $k=\vert \Q_d^{old}\vert$, there are $\prod_{i=i}^{n} \vert \fmap_i(f,\Q_d^{old})\vert$ iterations which is usually much smaller. Note that in many FTAs the size of the set $\fmap_i(f,\Q_d^{old})$ is  much smaller than $k$ (and is often zero) since the states of the input automaton tend to appear in only a few argument positions of function symbols.

\subsection{Tracking new values on each iteration}
We now apply an optimisation that further reduces the computation in the innermost loop. As it stands in Figure \ref{fig-alg4}, any state $Q_0$ generated on some iteration at line 17 is also generated on all subsequent iterations of the {\bf repeat} loop. To avoid this, we note that when evaluating the statement $Q_0 \gets \Rhs(\Delta_1 \cap \cdots \cap \Delta_n)$ in some iteration of the {\bf repeat} loop, a new value is obtained for $Q_0$ only when at least one of $\Delta_1, \ldots, \Delta_n$ is a new value, that is, one that was not available on the previous iteration. We therefore try to avoid re-evaluating old values of $\Delta_i$ for each $i$.

Some bookkeeping is needed to keep track of new values.  A variable $\Q_d^{new}$ represents the new elements of $\Q_d$ produced on some iteration. (The termination condition of the {\bf repeat} loop is altered to $\Q_d^{new}=\emptyset$). We introduce variables $\Psi^f_i$, which have the value of $\fmap_i(f,\Q_d)$.  The variables are initialised to $\emptyset$ and their values are augmented on each iteration. The statement $(\Phi_1,\ldots,\Phi_n) \gets (\fmap_1(f,\Q_d^{new})\setminus \Psi^f_1, \ldots, \fmap_n(f,\Q_d^{new})\setminus \Psi^f_n)$ computes the new sets of transitions from which the $\Delta_i$ sets can be chosen. 

The innermost {\bf for all} statement now iterates over the union of sets of tuples $Z$ where  $Z = \bigcup_{i=1}^n(\Psi^f_1\times \cdots \times \Psi^f_{i-1}\times \Phi_i  \times \Psi^{f,new}_{i+1} \times \cdots \times\Psi^{f,new}_n)$, which consists of exactly those tuples which contain at least one new value (that is, from one of the $\Phi_i$). Note in particular that if no new values for any argument are produced for some $f$ on some iteration, then the iteration set for $f$ on the next iteration is empty, since all the variables $(\Phi_1,\ldots,\Phi_n)$ have the value $\emptyset$.

\subsection{Computing the transitions of the DFTA in product form}\label{transitions}
As noted in Section \ref{sec-trans},  the set of transitions $\Delta_d$ of the determinised automaton can be computed from the final set of states $\Q_d$ with one extra iteration of the \textbf{repeat} loop.

\begin{figure}
  \centering

 \fbox{
\begin{minipage}{\textwidth}

\begin{algorithmic}[1]
       \State $\Delta_d \gets \emptyset$
	\ForAll {$f \in \Sigma$}
	\If {$(\arity(f)=0)$}
	  \State  $Q_0 \gets \Rhs(\func^{-1}(f))$
	  \If {$Q_0 \neq \emptyset$}
		   \State $\Delta_d\gets \Delta_d \cup \{f \to Q_0\}$
          \EndIf
	\Else
	\State $n \gets \arity(f)$
       	\State $\Psi^f_1, \ldots,\Psi^f_n \gets \fmap_1(f,\Q_d),\ldots,\fmap_n(f,\Q_d)$
	\ForAll {$(\Delta_1, \ldots \Delta_n) \in (\Psi^f_1\times \cdots \times \Psi^f_n)$}
	    	\State $Q_0 \gets \Rhs(\Delta_1 \cap \cdots \cap \Delta_n)$ 
		\If {$Q_0 \neq \emptyset$}
			\State $\Q_1, \ldots,\Q_n \gets \fmap^{-1}_1(f,\Q_d,\Delta_1),\ldots, \fmap^{-1}_n(f,\Q_d,\Delta_n)$
		 	\State $\Delta_d\gets \Delta_d \cup \{f(\Q_1,\ldots,\Q_n)\to Q_0\}$
		\EndIf 
	\EndFor
	\EndIf
	\EndFor

\end{algorithmic}
\end{minipage}
}
 \caption{Computation of product transitions}
  \label{fig-trans-3}
\end{figure}

\begin{figure}
  \centering
 \fbox{
\begin{minipage}{11 cm}
\begin{algorithmic}[1]
\Procedure{FTA Determinisation }{Input: $\langle Q,\Sigma,Q_f,\Delta\rangle$}
\State $\Q_d\gets \emptyset$
\ForAll {$f \in \Sigma$}
    \If {$(\arity(f)=0)$}
	  \State  $Q_0 \gets \Rhs(\func^{-1}(f))$
	  \If {$Q_0 \neq \emptyset$}
		   \State $\Q_d \gets \Q_d \cup \{Q_0\}$
          \EndIf
     \EndIf
     \State $(\Psi^f_1,\ldots,\Psi^f_n) \gets (\emptyset,\ldots,\emptyset)$ 
\EndFor
\State $\Q_d^{new} \gets \Q_d$
\Repeat
	\State $\Q_d^{old}\gets \Q_d$
	\ForAll {$f \in \Sigma$}
	       \If {$(\arity(f)>0)$}
	        \State $n \gets \arity(f)$
	        \State $(\Phi_1,\ldots,\Phi_n) \gets (\fmap_1(f,\Q_d^{new})\setminus \Psi^f_1, \ldots, \fmap_n(f,\Q_d^{new})\setminus \Psi^f_n)$
	        \State $(\Psi^{f,new}_1,\ldots,\Psi^{f,new}_n) \gets (\Psi^f_1 \cup \Phi_1,\ldots,\Psi^f_n \cup \Phi_n)$
	        \State $Z \gets \bigcup_{i=1}^n(\Psi^f_1\times \cdots \times \Psi^f_{i-1}\times \Phi_i  \times \Psi^{f,new}_{i+1} \times \cdots \times\Psi^{f,new}_n)$
                \ForAll {$(\Delta_1, \ldots \Delta_n) \in Z}$
		\State $Q_0 \gets \Rhs(\Delta_1 \cap \cdots \cap \Delta_n)$ 
		\If {$Q_0 \neq \emptyset$}
		   \State $\Q_d \gets \Q_d \cup \{Q_0\}$
		\EndIf 
	        \EndFor
	        \EndIf
	        	\State $(\Psi^f_1,\ldots,\Psi^f_n) \gets (\Psi^{f,new}_1,\ldots,\Psi^{f,new}_n)$
	\EndFor
        \State $\Q_d^{new} \gets \Q_d \setminus \Q_d^{old}$
\Until{$\Q_d^{new}=\emptyset$}
\State{Compute the set of transitions $\Delta_d$ (see Figure \ref{fig-trans-3})}
\State $\mc{Q}_{f} \gets \{Q' \in \Q_d \;|\; Q' \cap Q_f \neq\emptyset\}$
\State \Return $(\Q_d, \Sigma, \mc{Q}_{f}, \Delta_d)$
\EndProcedure
\end{algorithmic}
\end{minipage}
}
 \caption{Optimised algorithm}
  \label{fig-alg6}
\end{figure}
The idea is that, at the point in Figure \ref{fig-alg4} where state $Q_0$ is generated (line 17), we could have generated a transition $f(Q_1,\ldots,Q_n) \rightarrow Q_0$. 
However, following the optimisations of the inner loop described in Section \ref{step5}, we no longer compute the tuple $(Q_1,\ldots,Q_n)$ explicitly. Instead, we keep track of the information needed to generate the tuple $(Q_1,\ldots,Q_n)$ and the function $f$, so that we can construct the transition later.

To do this efficiently, we introduce another function $\fmap^{-1}_i: (\Sigma \times 2^{2^Q}\times  2^{\Delta})  \rightarrow 2^{2^Q}$, defined as follows.

 \[
\fmap^{-1}_i(f,\Q,\Delta') = \{Q' \mid \Lhs_i(f,Q')=\Delta', Q' \in \Q\}
\]

The set of transitions for $f$ is then all transitions of the form $f(Q_1,\ldots,Q_n) \rightarrow \Rhs(\Delta_1 \cap \cdots \cap \Delta_n)$ such that 
$(\Delta_1, \ldots \Delta_n) \in (\fmap_1(f,\Q_d)\times \cdots \times\fmap_n(f,\Q_d)$ and
$(Q_1,\ldots,Q_n) \in (\fmap^{-1}_1(f,\Q_d,\Delta_1)\times \cdots \times \fmap^{-1}_n(f,\Q_d,\Delta_n)$.

Product transitions for $\Delta_d$ can then be obtained directly. We simply omit the enumeration of the tuples $Q_1,\ldots,Q_n$  which are  elements of $\fmap^{-1}_1(f,\Q_d,\Delta_1)\times \cdots \times \fmap^{-1}_n(f,\Q_d,\Delta_n)$. This gives the algorithm in Figure \ref{fig-trans-3} for generating $\Delta_d$ in product form.  
In the expression $f(\Q_1,\ldots,\Q_n)\to Q_0$ in Figure \ref{fig-trans-3}, $\Q_1,\ldots,\Q_n$ are sets of DFTA states.  

\subsection{Further optimisation}
The final algorithm then consists of Figure \ref{fig-alg6} for computing $\Q_d$ together with Figure \ref{fig-trans-3} for computing $\Delta_d$ in product form.
We note that in Figure \ref{fig-trans-3} for computing $\Delta_d$, the values of $\Psi^f_1, \ldots,\Psi^f_n$ are available from the main {\bf repeat} loop in Figure \ref{fig-alg6} and do not need to be recomputed.  Also, the values of the expressions $\fmap^{-1}_i(f,\Q_d,\Delta_i)$ can be computed in the main loop of Figure \ref{fig-alg6} by tabulating computed values of $\Lhs_i$;  more precisely, whenever an expression $\Lhs_i(f,Q')$ is evaluated and yields a non-empty value $\Delta'$, $Q'$ is added to the set of values $\fmap^{-1}_i(f,\Q_d,\Delta')$. 
%

 %
\section{Performance of the optimised algorithm}\label{complexity}

\paragraph{Worst Case} Consider first the worst case running time for determinising the FTA $\langle Q, F, \Sigma, \Delta\rangle$.  The size of the input is measured by $|\Sigma|$, $|Q|$ and $n$, the maximum arity of the elements of $\Sigma$. 

The main {\bf repeat} loop of the algorithm in Figure \ref{fig-alg6} can be traversed up to $2^{|Q|}$ times, which is the upper bound of the number of states in the DFTA.  For each $f \in \Sigma$ within the main loop there are up to $(2^{|Q|})^n$ iterations, since the size of $\fmap_i(f,\Q^{old}_d)$ can be up to $|\Q_d|=2^{|Q|}$.  Combining all three nested loops, the complexity of the main loop of the algorithm is $O(\vert \Sigma \vert . 2^{|Q|.(n+1)})$.

Considering the product transitions, the maximum number of iterations of the transition-generation loop for each $n$-ary $f \in \Sigma$ is $(2^{|Q|})^n$ and there is at most one product transition generated in each iteration. Hence the number of product transitions generated in the worst case is $O(\vert \Sigma \vert . 2^{n.|Q|})$. If the automaton is also completed, as discussed in Section \ref{completion}, then $|Q|$ is possibly increased by one (the state $\any$).

Regarding both the running time and the size of the output, the optimised algorithm performs no better in the worst case than the textbook algorithm with explicit enumeration of the DFTA transitions. 

\paragraph{Running time in practice} The worst case of $2^{|Q|}$ for the number of DFTA states seems to be approached only in unusual situations; to achieve it, for every pairs of states $q,q'$ of the original automaton it would have to be the case that $L(q) \cap L(q')$, $L(q) \setminus L(q')$ and $L(q') \setminus L(q)$ are all nonempty (since the states of the generated DFTA include only ones accepting some term, see Lemma \ref{dfta-nonempty}).  For such a pair of states, it is more common either that $L(q)$  and $L(q')$ are disjoint or that one includes the other.  If this is the case for all such pairs, then the number of DFTA states is at most the number of original states.  In Example \ref{ex2} the number of DFTA states is the same as the number of original states; our experiments (Section \ref{experiments}) show a tendency for the size of the set of states of the FTA and the  corresponding DFTA to be close for examples where the FTA is a description of program data.  The set of DFTA states may even be smaller, if the original FTA has redundant or duplicated states - which sometimes happens with automatically generated FTAs.

Replacing $2^{|Q|}$ by $\vert\Q_d\vert$ in the complexity expressions gives a running time of $O(\vert \Sigma \vert . \vert\Q_d \vert^{n+1})$ and $O(\vert \Sigma \vert . |\Q_d|^{n})$ for the number of product transitions in the output. Even if $|\Q_d|$ is small, we can see that high-arity function symbols still present a potential for blow-up. Again, in practice this danger is often greatly reduced by the structure of the input FTA.  
As already noted, the size of $\fmap_i(f,\Q_d)$, whose worst-case size is $|\Q_d|$ is usually much smaller than $|\Q_d|$.  This is due to the natural ``typing'' of function symbols.  A function argument position in the original FTA is typically associated with a small number of states.  However, there are applications where the size of $\fmap_i(f,\Q_d)$ is larger and for these is a danger of blow-up for high-arity function symbols.  Such applications will be discussed in Section \ref{experiments}.
%
%
 %
\section{Complete DFTAs }\label{completion}

Recall that in a complete FTA (Definition \ref{complete}) every term $t \in\Term(\Sigma)$ has a run $t \Rightarrow^* q$ where $q \in Q$. An FTA can always be completed \cite{Comon}, by adding an extra state to $Q$ and adding extra transitions to $\Delta$.  
\begin{example}\label{ex-completion}
Consider the following FTA $\langle Q,\Sigma, Q_f,\Delta\rangle$ where 
$Q = \{\listtype,\num\}$, $Q_f = \{\listtype\}$
$\Sigma = \{[], [.|.], 0,s(.)\}$,  and $\Delta = 
\{[] \rightarrow \listtype, [\num | \listtype]  \rightarrow \listtype,
0 \rightarrow \num, s(\num)  \rightarrow \num \}$.  The FTA is not complete; for instance there is no transition with left hand side $s(\listtype)$ or $[\num|\num]$.  Thus the terms $s([])$ and $[s(0)|0]$, for example, have no run. To complete it, we can add an extra non-final state, say $\e$ (for \emph{error}), and add the following transitions.
\[
\begin{array}{ll}
\Delta_{e} = &\{s(\listtype) \to \e, s(\e) \to \e, [\listtype|\listtype] \to \e,\\ 
& [\num|\num] \to \e, [\listtype|\num] \to \e,[\e|\listtype] \to \e,\\
& [\listtype|\e] \to \e, [\e|\e] \to \e, [\num|\e] \to \e, [\e|\num] \to \e \}
\end{array}
\]
\noindent
The FTA $\langle Q \cup \{\e\},\Sigma, Q_f,\Delta \cup \Delta_e\rangle$ is complete and accepts the same language as the original FTA (the set of lists of numbers).  Any term that did not have a run in the original now has a run to $e$ (which is not an accepting state).
\end{example}
One can see that the number of transitions in the completed FTA is determined by the arity of the function symbols and the number of states, and that it is exponential in the arity of the function symbols.  

It would be possible to modify the textbook determinisation procedure in Figure \ref{fig-alg1} to complete the FTA as well as determinise it.  However, it would then be more involved to apply the optimisations and generate product transitions.  The approach given in the next section is to add transitions to the input FTA (namely $\Delta_{\any}$) that guarantee that the resulting DFTA is complete.  The optimisations introduced in the previous sections then take effect, and transitions are generated in product form, without any further modification of the algorithm. 

\subsection{Simultaneous completion and determinisation using $\Delta_{\any}$}

Recall that given a finite signature $\Sigma$ we define the set of transitions $\Delta_{\any}$ as
$\{f(\stackrel{n {\rm ~times}}{\overbrace{\any,\ldots,\any}}) \rightarrow \any ~|
f^n \in \Sigma\}$.  Clearly for any term $t \in \Term(\Sigma)$ there exists a run $t \Rightarrow^* \any$. The  following lemma shows that we can use $\Delta_{\any}$ to obtain a complete DFTA from a given FTA, in other words, we can perform determinisation and completion simultaneously.  First we establish an important property of the DFTAs generated by the algorithm - that they contain only ``productive'' states.

\begin{lemma}\label{dfta-nonempty}
Given an FTA, for every state $Q'$ of the DFTA obtained by determinising it using the given (optimised) algorithm,  $L(Q') \neq \emptyset$.
\end{lemma}

\begin{proof}
We consider the algorithm in Figure \ref{fig-alg4} for simplicity, rather than the final optimised algorithm. Consider a DFTA state $Q'$ that is generated in the algorithm.  We reason by induction on the number of  iterations of the \textbf{repeat} loop of the algorithm. 
\begin{itemize}
\item
Base case $(i=0)$: If $Q'$ is generated before the first iteration, then there exists some 0-arity function $f$ such that $Q' = \Rhs(\func^{-1}(f))$.  The transition $f \to Q'$ is generated on iteration 1 hence there is a run $f \Rightarrow Q'$. 
\item
Induction $(i > 1)$:  Assume that the lemma holds for all states generated up to the $i-1^{th}$ iteration, and that $Q'$ is a new state generated on the $i^{th}$ iteration.  Then there exist states $Q_1,\ldots,Q_n$ which were generated in the first $i-1$ iterations, such that a transition $f(Q_1,\ldots,Q_n) \to Q'$ is constructed on the $i^{th}$ iteration.  By the inductive hypothesis, there exist terms $t_1,\ldots,t_n$ such that $t_i \Rightarrow^* Q_i, 1 \le i \le n$. Hence there is a run $f(t_1,\ldots,t_n) \Rightarrow^* Q'$.
\end{itemize}
Hence the lemma holds for states generated on any iteration.
\qed
\end{proof}

\begin{lemma}\label{dfta-complete}
Let $\langle Q, Q_f, \Sigma, \Delta\rangle$ be an FTA such that every term $t \in \Term(\Sigma)$ has a run $t \Rightarrow^* q$ for some $q \in Q$. Then the DFTA obtained by determinising $\langle Q, Q_f, \Sigma, \Delta \rangle$ is complete.
\end{lemma}

\begin{proof}
Let $\Q'$ be the set of states of the generated DFTA, and assume that it is not complete. Then there exist $Q_1,\ldots,Q_n \in \Q'$ and an $n$-ary function $f$ such that for all $Q_0$ there is no transition $f(Q_1,\ldots,Q_n) \to Q_0$ in the DFTA.  Let $t_1,\ldots,t_n \in \Term(\Sigma)$ be terms such that $t_i \Rightarrow^* Q_i, 1 \le i \le n$, whose existence is guaranteed by Lemma \ref{dfta-nonempty}.  Since it is a DFTA, there cannot be any other runs $t_i \Rightarrow^* Q_i^{\prime}$ where $Q_i \neq Q_i^{\prime}$ respectively.  Hence there is no run for the term $f(t_1,\ldots,t_n)$, which contradicts the assumption of the lemma.
Hence the generated DFTA is complete.
\qed
\end{proof}
Thus we establish the following.
\begin{proposition}\label{any-complete}
Let $\langle Q,Q_f,\Sigma,\Delta \rangle$ be an FTA, such that $\any \in Q$ and $\Delta_{\any} \subseteq \Delta$.  Then the DFTA  obtained by the determinisation algorithm with this input is complete.
\end{proposition}

\begin{proof}
For every term $t \in \Term(\Sigma)$ the FTA has a run $t \Rightarrow^* \any$.  The result follows from Lemma \ref{dfta-complete}.
\qed
\end{proof}

\subsection{Performance of the algorithm after adding $\Delta_{\any}$ }
Although by Proposition \ref{any-complete} we can obtain a complete DFTA from any given input FTA, simply by adding the state \any\ and the transitions $\Delta_{\any}$ to the input before running the algorithm, we may ask what is the impact on the performance of the determinisation algorithm.  

The impact on the number of states of the DFTA is slight; at most one extra state $\{\any\}$ is generated, in the case that there are some terms accepted by \any\ but not by any other state.  This state represents the ``error'' state of the classical completion procedure. Apart from this, the same states are generated but \any\ is added to each one;  it is easy to see that \any\ must appear in every DFTA state since the state \any\ appears in every possible left-hand-side position in $\Delta_{\any}$.  

With standard transitions, completion can cause a huge increase in the size of $\Delta_d$.
The main question is thus the impact on the product representation of $\Delta_d$.  
Let us analyse the effect of adding $\Delta_{\any}$ on part of the algorithm generating the transitions of the DFTA in product form (Figure \ref{fig-trans-3}).
Consider the effect of the introduction of $\Delta_{\any}$ on the basic operations of the algorithm.

\begin{itemize}
\item
$\func^{-1}(f)$ and $\Lhs_i(f,Q)$:  in each case the returned set contains at most one extra transition (that is, the transition $f(\any,\ldots,\any) \to \any$, in the case that $\any \in Q$).
\item
$\Rhs(T)$: the returned set contains at most one extra state \any\ in the case that $T \cap \Delta_{\any} \neq \emptyset$.
\end{itemize}
  
The worst case of the number of extra iterations of the inner loop in Figure \ref{fig-trans-3} caused by the extra state is discussed in Section \ref{complexity}.

In Section \ref{experiments} we verify that the time overhead of adding $\Delta_{\any}$ is small.  The overhead in the number of product transition is discussed in Section \ref{experiments}, and is in practice a small factor. Thus we obtain the completed DFTA at little extra cost over obtaining the DFTA.

\subsection{Don't-Care Arguments in Complete DFTAs}\label{dontcare}

An underscore ``\_'' is used as shorthand for $\Q_d$ in a product transition; that is, $f(\Q_1,\ldots,\_, \ldots,\Q_n) \to Q_0$ is the same as $f(\Q_1,\ldots,\Q_d, \ldots,\Q_n) \to Q_0$. This indicates that the choice of DFTA state in this argument is irrelevant in determining the right hand side $Q_0$. Product transitions of the form $f(\_,\ldots,\Q_i,\ldots, \_) \to Q_0$ in which all but one argument are don't-care arguments are especially interesting, since the right-hand-side of the transition is determined by just one argument.  We call the elements of such arguments \emph{deciding arguments}.  

A typical case of deciding arguments arises in complete DFTAs constructed by adding $\Delta_{\any}$ to the original FTA, where a state $\{\any\}$ is generated, which accepts the terms not accepted by any other state.  $\{\any\}$ is a deciding argument; the presence of $\{\any\}$ in any argument in a DFTA transition is sufficient to ensure that the right hand side of the transition is $\{\any\}$. That is, there are DFTA product transitions of the form:
\[
\begin{array}{c}
f(\{\{\any\}\},\_,\ldots,\_,\_) \to \{\any\}\\
f(\_,\{\{\any\}\},\_,\ldots,\_) \to \{\any\}\\
\ddots\\
f(\_,\_,\ldots,\_,\{\{\any\}\}) \to \{\any\}\\
\end{array}
\]
\noindent
These product transitions overlap, obviously, since $\{\any\}$ is included in the don't-care arguments.  However, this form might be much more compact than the product transitions generated by the determinisation algorithm.  Furthermore, there could be other deciding arguments besides $\{\any\}$.

We prove two lemmas defining sufficient conditions for finding deciding arguments and generating the corresponding don't-care product transitions.

\begin{lemma}\label{deciding} Let $\Q_d$ be the set of states of a compete DFTA and let
 $\Psi_1,\ldots,\Psi_n = \fmap_1(f,\Q_d),\ldots,  \fmap_n(f,\Q_d)$ for some $n$-ary function $f$. Let $\Delta' \in \Psi_i$ and $\Q_i = \fmap^{-1}_i(f,\Q_d,\Delta')$.  Then
$\Q_i$ are \emph{deciding arguments} for the $i^{th}$ argument of $f$ if
\[
 \Rhs(\Delta'\cap \bigcap(\cap \Psi_1,\ldots,\cap \Psi_{i-1},\cap \Psi_{i+1},\ldots,\cap \Psi_n)) = \Rhs(\Delta').
\]

\end{lemma}

\begin{proof}

Consider the set of right hand sides of transitions that can be built from $(\Psi_1\times \cdots \times \Psi_n)$ using $\Delta'$ in the $i^{th}$ position, say $\bar{Q}$. That is, $\bar{Q} = \Rhs(\Delta_1 \cap \cdots \cap  \Delta' \cap \cdots \cap \Delta_n)$ where $(\Delta_1, \ldots,\Delta',\ldots, \Delta_n) \in (\Psi_1\times \cdots \times \{\Delta'\} \times \cdots \Psi_n)$ with $\Delta'$ in the $i^{th}$ position.

$\Rhs(\Delta')$ is an upper bound for $\bar{Q}$, since $\Delta_1 \cap \ldots,\Delta' \cap \ldots \cap \Delta_n \subseteq \Delta'$ and $\Rhs$ is monotonic.  The expression 
$ \Rhs(\Delta'\cap \bigcap(\cap \Psi_1,\ldots,\cap \Psi_{i-1},\cap \Psi_{i+1},\ldots,\cap \Psi_n))$ is a lower bound for $\bar{Q}$, since $\Delta'\cap \bigcap(\cap \Psi_1,\ldots,\cap \Psi_{i-1},\cap \Psi_{i+1},\ldots,\cap \Psi_n) \subseteq \Delta_1 \cap \cdots \cap  \Delta' \cap \cdots \cap \Delta_n$ and $\Rhs$ is monotonic.  If these two are equal, as in the statement of the property, then we can conclude that the value $\Rhs(\Delta')$ is the right hand side for any such transition since it is both an upper and a lower bound. 

\qed
\end{proof}
If we can find such a $\Delta' \in \fmap_i(f,\Q_d)$, a (product) transition $f(\ldots,\_,\Q_i,\_,\ldots) \to \Rhs(\Delta')$ is constructed, where $\Q_i=\fmap^{-1}_i(f,\Q_d,\Delta')$ and the underscore arguments stand for $\Q_d$.

A more specialised sufficient condition for deciding arguments for binary functions is given by the following lemma.

\begin{lemma}\label{deciding2} Let $\Q_d$ be the set of states of a complete DFTA and let
 $\Psi_1,\Psi_2 = \fmap_1(f,\Q_d), \fmap_2(f,\Q_d)$ for some $2$-ary function $f$. Then a 
 set of DFTA states $\Q_i \subseteq \Q_d$, where $\Q_i = \fmap^{-1}_i(f,\Q_d,\Delta')$ for some $\Delta' \in \Psi_i$, $i \in \{1,2\}$ are \emph{deciding arguments} for the $i^{th}$ argument of $f$ if $\Rhs(\Delta')$ is a singleton, and for all $\Delta'' \in \Psi_j$, $j \in \{1,2\} \setminus \{i\}$,  
\[
\Delta'\cap \Delta'' \neq \emptyset.
\]

\end{lemma}

\begin{proof}
$\Delta'\cap \Delta'' \neq \emptyset$ implies that $\Rhs(\Delta_j \cap  \Delta') \neq \emptyset$.  Since $\Rhs(\Delta')$ is a singleton and $\Rhs(\Delta'' \cap  \Delta') \subseteq \Rhs(\Delta')$,  $\Rhs(\Delta'' \cap \Delta') = \Rhs(\Delta')$.  
\end{proof}
If such a $\Delta'$ is found, say in argument 2, we generate the product transition $f(\_,\Q_2) \to \Rhs(\Delta')$ where $\Q_2=\fmap^{-1}_2(f,\Q_d,\Delta')$ and the underscore arguments stand for $\Q_d$.

We can easily add a check for deciding arguments using the conditions of Lemma \ref{deciding} and \ref{deciding2} to the algorithm, just before generating product transitions. The calculation of the intersections is exponential in the arity of the function symbols, but does not alter the complexity of the overall algorithm and can save effort in generating product transitions. For each set of deciding arguments $\Q_i$ discovered, we  generate a don't-care product transition of the form just shown, 
and the corresponding value $\Delta'$ is removed from $\Psi_i$ when computing the remaining product transitions for $f$.

The problem of finding the minimum number of product transitions to represent the DFTA transitions seems to be intractable and is beyond the scope of this paper.  In essence it can be stated as the problem of finding the minimum number of cartesian products whose union is a given relation. 
 %
\section{Experiments}\label{experiments}

Tables \ref{results1} and \ref{results2} show experimental results comparing the optimised determinisation algorithms generating transitions in product form with the textbook algorithm.  It also compares the effect of adding the detection of don't care arguments in the determinisation algorithm.  The algorithms are implemented in Java;\footnote{The code is available at \url{https://github.com/bishoksan/DFTA}} the textbook algorithm is a fairly direct implementation of the program in Figure \ref{fig-alg4}. Our own implementation of the textbook algorithm could be improved but the same overall performance is likely in terms of the number of solvable problems, since the size of the output set of transitions is often the limiting factor. In Section \ref{comparison}, we compare with other implementations of FTA operations, where we find that our implementation of the textbook algorithm seems comparable in performance with existing determinisation tools.

The 14,694 benchmark FTAs were obtained from the repository that is part of libvata.  Many of these FTAs originate in the Timbuk system \cite{GenetT01} and others from the abstract regular tree model checking experiments (ARTMC) \cite{DBLP:conf/sas/BouajjaniHRV06}.  All experiments were carried out on a computer running Linux with 4-core Intel\textsuperscript{\textregistered} Xeon\textsuperscript{\textregistered} X5355 processors running at 2.66GHz.

\subsection{Determinisation and completion}

The columns in Table \ref{results1} show the overall effectiveness of three versions of the determinisation algorithm, for determinisation and determinisation with completion ({\rm +compl}).  DFTA is the textbook algorithm generating all transitions explicitly, DFTA-opt is the optimised algorithm returning product form, without detection of don't care arguments, and DFTA-opt-dc is with detection of don't cares.

\begin{table}[t]
\newcolumntype{L}[1]{>{\raggedright\arraybackslash}p{#1}}
\newcolumntype{C}[1]{>{\centering\arraybackslash}p{#1}}
\newcolumntype{R}[1]{>{\raggedleft\arraybackslash}p{#1}}

\begin{tabularx}{\textwidth}{|L{0.13\textwidth}|*{6}{R{0.1028\textwidth}|}}
\hline
~& \multicolumn{2}{c|}{\rm DFTA} & \multicolumn{2}{c|}{ \rm DFTA-opt}  &\multicolumn{2}{c|}{ \rm DFTA-opt-dc}        \\ 
 \hline
~			&      			& {\rm +compl~} &     		& {\rm +compl~} & 		&  {\rm +compl~ }          \\ \hline
solved             	& 112      		& 109         	& 14672        	& 14670         	&14672	& 14669      	\\ \hline
t/o           		& 14584    	& 14587       	& 22           	& 24              	&22      	&  25    		 \\ \hline
avg. secs. 	& 119.9     	& 119.9        	& 0.30         	& 0.37        	&0.30	&  0.34       	\\ \hline
\% solved          & 0.76\%     	& 0.74\%    	& 99.85\%       	& 99.84\%        	&99.85\% 	&99.83\%        	\\ \hline
\end{tabularx}
    \vskip 0.5cm
\caption {Average time (in seconds) for determinisation of 14,694 benchmark problems, with and without completion (timeout 120 seconds). DFTA = textbook algorithm; DFTA-opt = optimised algorithm without don't care detection;  DFTA-opt-dc =  optimised algorithm with don't care detection}\label{results1}
\end{table}

The first notable point is that the textbook algorithm is able to solve less than 1\% of the problems while the optimised algorithms solve nearly all of them.  The running time of the textbook algorithm is far slower even considering only those problems that it could solve.

The size of the input and output DFTAs is summarised in Table \ref{results2}. The size of the set of states and transitions of the input FTA are $|Q|$ and $|\Delta|$.  The number of DFTA states is $|Q_d|$ and the number of product transitions in the DFTA is $|\Delta_{\Pi}|$.  For completed DFTAs,  the precise size of the set of transitions $\Delta_d$ depends only on the signature and can be calculated; this is shown in the table as $|\Delta_d|$.  For non-complete DFTAs we cannot usually directly enumerate $\Delta_d$ since it is too large. However, we estimate its size by summing the product of the sizes of the product states in each transition.  (However, the same transition could be represented by more than one product transition so it is an over-estimate and we show it in brackets in Table \ref{results2}). The comparison of the sizes and their distribution of $\Delta_d$ and $\Delta_{\Pi}$ is shown in the scatter-plots in Figure \ref{fig-chart2}. Note that the scale on the vertical axes is logarithmic, and that while the majority of the benchmarks have $\Delta_d$ with size over $10^{12}$, the sizes of $\Delta_{\Pi}$ lie between 10 and $10^4$.  This is a dramatic difference.

It can immediately be seen that the number of DFTA states is on average only slightly greater than the number of input FTA states, as discussed in Section \ref{complexity}. The average size of the set of transitions in completed DFTAs is extremely large, and shows immediately why the textbook algorithm fails.  Regarding completion, the size of the set of states of both input and output automata is increased by one.  The size of the input set of transitions is increased by the size of $\Delta_{\any}$ which is $|\Sigma|$.

Determinisation with completion usually results in a larger set of product transitions than determinisation alone.  Without don't cares, the average increase is a factor of about 9 (13,908 compared to 1,563). With don't cares, the average increase factor is only about 2.5. However, in the optimised algorithms, completion hardly increases the running time, nor is detection of don't cares a significant overhead.

\begin{table}[t]
\newcolumntype{L}[1]{>{\raggedright\arraybackslash}p{#1}}
\newcolumntype{C}[1]{>{\centering\arraybackslash}p{#1}}
\newcolumntype{R}[1]{>{\raggedleft\arraybackslash}p{#1}}

\begin{tabularx}{\textwidth}{|L{0.189\textwidth}|*{4}{R{0.1575\textwidth}|}}
\hline
~&  \multicolumn{2}{c|}{ \rm DFTA-opt}   &\multicolumn{2}{c|}{ \rm DFTA-opt-dc}        \\ 
\hline
~						&            		& {\rm +compl~} 		& 		&  {\rm +compl~ }          	\\  \hline
average  $|Q|$	                 	& 58.67     	& 58.67           			& 58.67    	& 58.52         			\\ \hline
average  $|\Delta|$            	& 293.58 		& 307.66       			& 293.58 	& 306.23      			 \\ \hline
average  $|\Sigma|$            	& 15.5 		& 15.5 				& 15.5 	& 15.5    				 \\ \hline
average  $|Q_d|$	                 & 110.10     	& 102.08           		& 110.18   & 96.55           			 \\ \hline
average  $|\Delta_d|$           	& ($2.23\times 10^{6}$)& $2.93\times 10^{18}$      & ($2.25\times 10^{6}$)	& $2.93 \times  10^{18}$      \\ \hline
average  $|\Delta_{\Pi}|$     	& 1563.13   	& 13908.79       		& 1565.58   & 3817.90         			\\ \hline
\end{tabularx}
    \vskip 0.5cm
\caption {Size statistics for input and output of optimised DFTA algorithms (solved problems only).}\label{results2}

\end{table}

\begin{figure}
\begin{center}
\includegraphics[width=3.5 in]{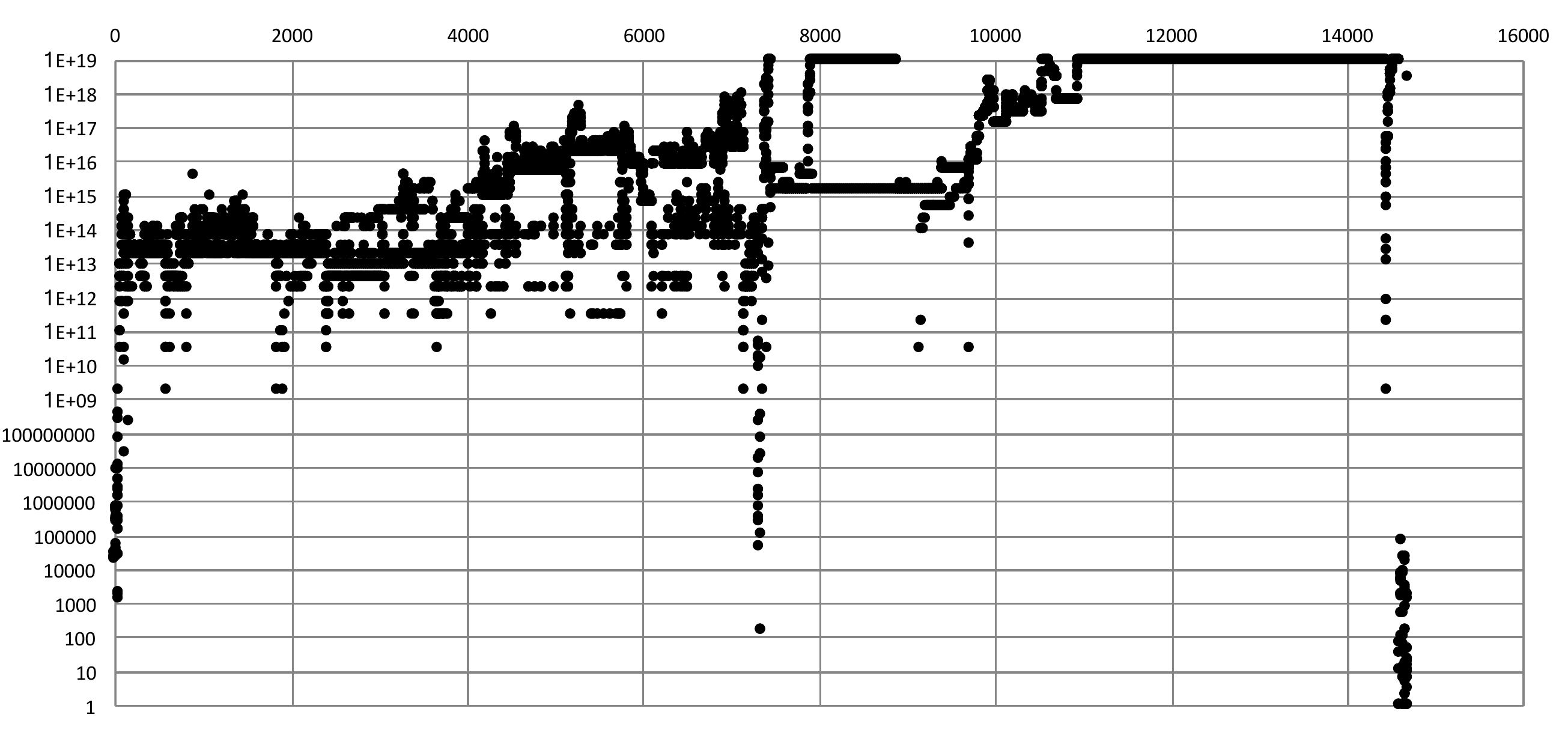}
\includegraphics[width=3.5 in]{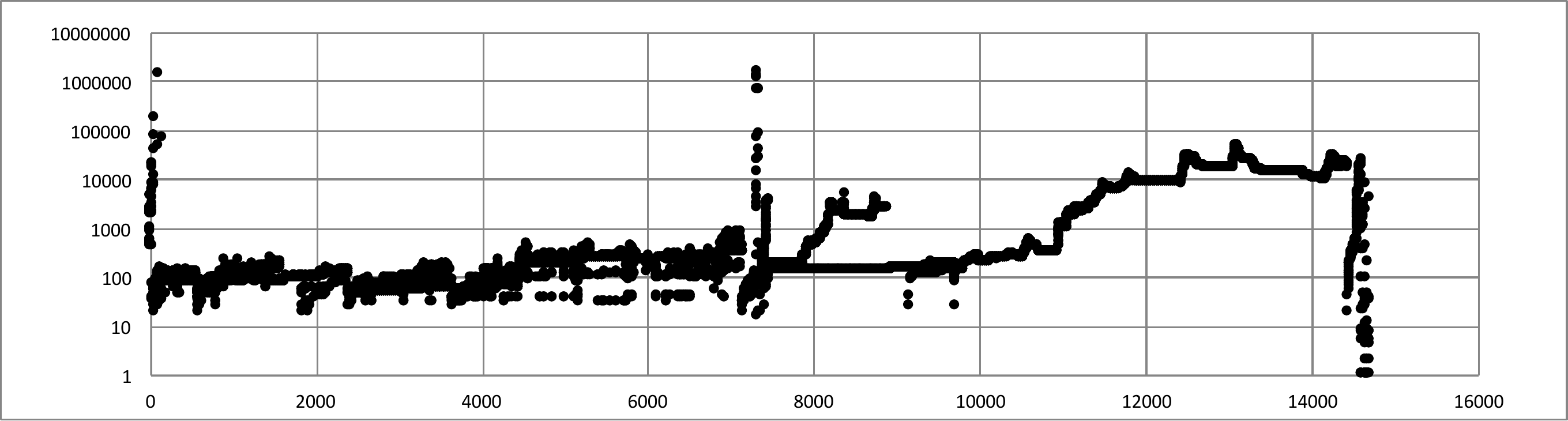}
\end{center}
\caption{Distribution of size of $\Delta_d$ and $\Delta_{\Pi}$  for 14,672 complete determinised automata. }
\protect\label{fig-chart2}
\end{figure}

\subsection{Comparison with other FTA implementations}\label{comparison}

We have investigated the implementation of determinisation and completion (or closely related operations) in other available FTA libraries, namely Timbuk\footnote{\url{http://people.irisa.fr/Thomas.Genet/timbuk/TimbukOnLine/}} \cite{GenetT01}, MONA\footnote{\url{http://www.brics.dk/mona/}} \cite{MONA14}, the VATA tree automaton library\footnote{\url{http://www.fit.vutbr.cz/research/groups/verifit/tools/libvata/}} \cite{LengalSV12} and Autowrite\footnote{\url{http://dept-info.labri.fr/~idurand/trag/}} \cite{DBLP:journals/entcs/Durand05}.

Timbuk provides a determinisation operation, which however is not optimised since determinisation plays no role in the main applications areas of Timbuk, which are mainly the verification of cryptographic protocols\footnote{T. Genet. Personal communication}. We have not made an extensive comparison since the implementation seem to perform at best the same as our textbook implementation. Using the online version of Timbuk, the complementation of a medium-sized example with $|Q|=53$ and $|\Delta|=174$ runs out of resources, while our textbook implementation takes 7.5 seconds to produce the result with $|\Delta_d|=23,535$.
On the other hand, the optimised algorithm DFTA-opt computes the result in 50 milliseconds, yielding 501 product transitions.

MONA uses a technique called ``guided tree automata'' to optimise operations on tree automata. Guided tree automata are applied to FTAs that are already determinised by the user, and thus MONA does not explicitly offer a determinisation operation. Certain operations in MONA (for example projection) make the automata temporarily non-deterministic, and so these operations incorporate a determinisation operation to restore deterministic form. However, this operation handles only the restricted form of non-deterministic FTA that can arise, and is not directly accessible for testing, and so we are unable to make any direct comparison of our algorithm with MONA. 

We made more extensive comparisons with VATA and Autowrite.
The VATA library contains an FTA complementation operation.  Complementation in VATA is not carried out by the classical procedure of determinisation and completion, followed by switching the accepting and non-accepting states of the output DFTA. However, it is the closest comparable operation. We took a set of 93 FTAs from the ARTMC system (Abstract Regular Tree Model Checking) benchmarks from the VATA benchmark library (which were included in the 14,672 previously mentioned). This is the set in which most of the timeouts of our optimised algorithm occur and therefore provide the most challenging comparison.   Figure \ref{fig-chart3} shows the results of running VATA,  DFTA and DFTA-opt (with completion) on these, with a timeout of 120 seconds.  It can be seen that DFTA-opt solves about 75\% of these, the majority of them fast, while VATA succeeds on only 5 within the time limit, comparable to our implementation of the textbook algorithm, which solves 15 within the time limit.   In summary, from this sample it appears that VATA's complementation performs at best as well as the textbook algorithm, but worse than the optimised algorithm for determinisation with completion.

The Autowrite tool is part of the TRAG system providing a library for analysis of term rewriting systems and graphs.  Tree automata operations play an important role and thus the implementation of tree automata has been done with care \cite{DBLP:journals/entcs/Durand05}. The operations include determinisation and complementation.   However, complementation does not generate all the transitions explicitly and the computation is essentially the same as determinisation. Therefore we compared with our optimised algorithm for determinisation without completion. On the same set of 93 problems used for comparison with VATA, Autowrite's determinisation performance is comparable to our implementation of the textbook algorithm, solving 17 out of 93 within the 120 second time limit. The optimised algorithm for determinisation succeeded in 83 of the examples within the limit, with an average time (for the solved problems) of 17 seconds.

\begin{figure}
\begin{center}
\includegraphics[width=3 in]{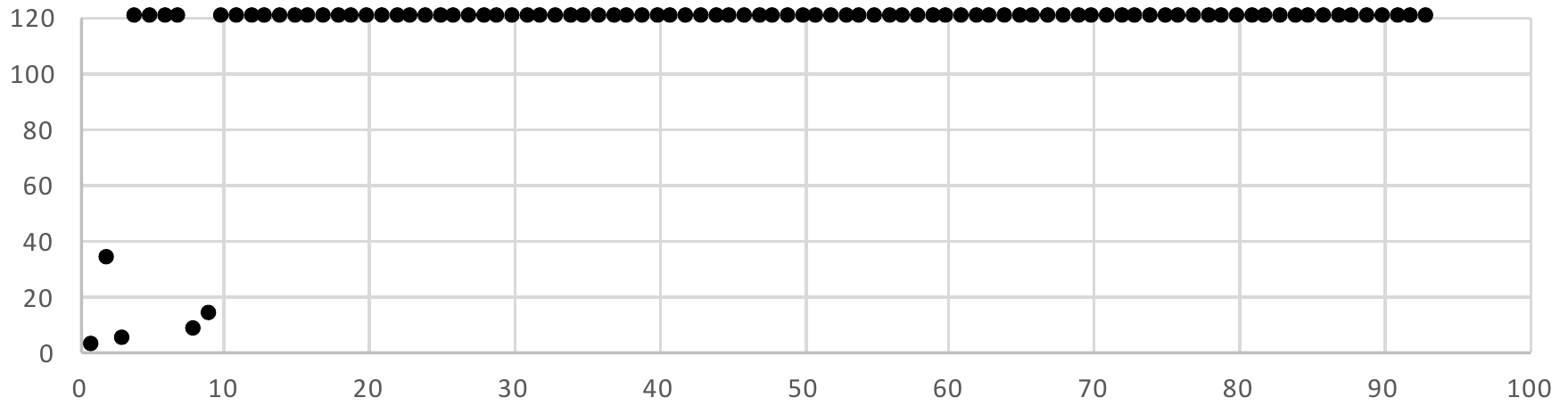}
\vskip 0.5cm
\includegraphics[width=3 in]{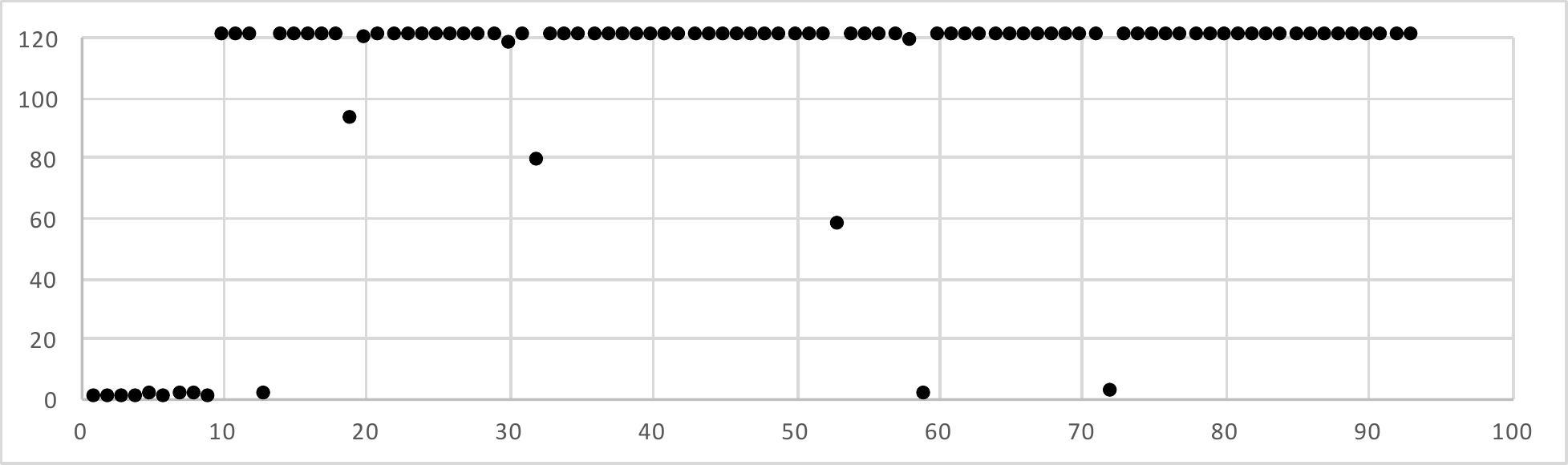}
\vskip 0.5cm
\includegraphics[width=3 in]{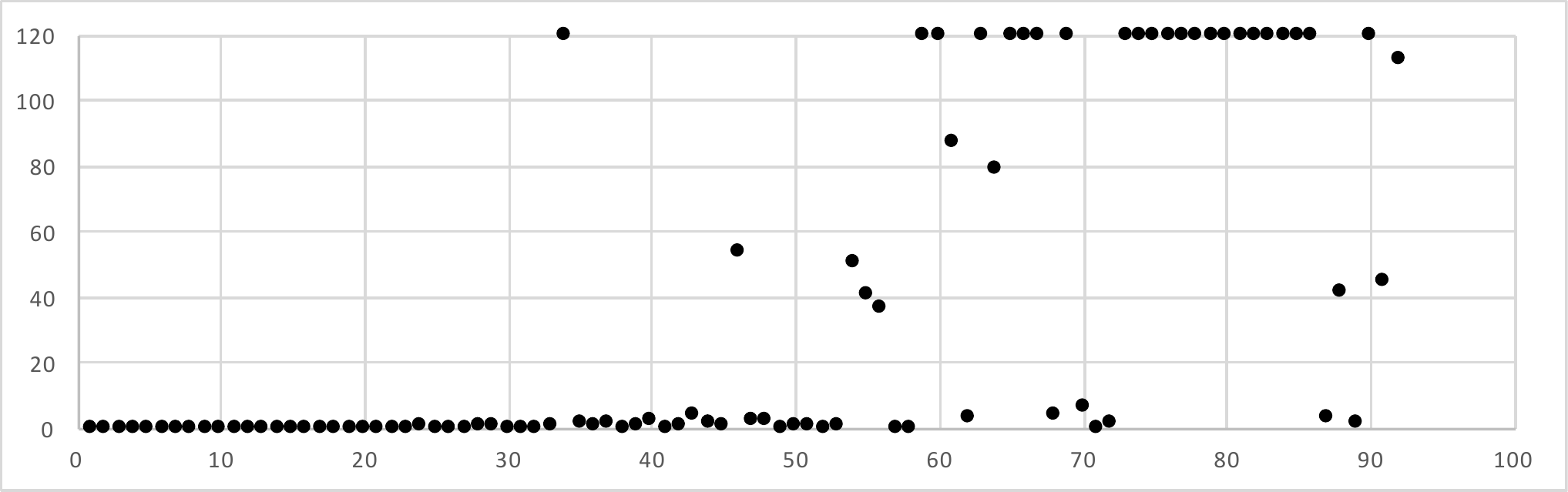}
\end{center}
\caption{Comparison of time in seconds for complementation of 93 ARTMC benchmarks (timeout 120 secs.) with VATA (upper), DFTA-textbook (middle) and  DFTA-opt (lower) }
\protect\label{fig-chart3}
\end{figure}

\subsection{Performance on random FTAs}

Figure \ref{fig-chart5} shows the result of running DFTA-opt on a set of 360 random FTAs.\footnote{Obtained using a generator kindly supplied by R. Mayr and R. Almeida}
These FTAs have 20 states, function symbols with arity 0, 2, 3 or 4, a transition density (the average number of different right-hand side states for a given
left-hand side of a transition rule) ranging between 1.5 and 1.7 and an acceptance density (proportion of accepting states) ranging from 0.5 to 0.8. These were more challenging examples for all the algorithms tested. DFTA-opt solved 11.1\% within the timeout as shown in Figure \ref{fig-chart5}. The result for VATA's complementation are not shown, but only 1.1\% of the FTAs were complemented by VATA within the timeout.

The reason for the relatively poor performance of our algorithm (though still much better than the state of the art) on random FTAs is that the optimised algorithm's success depends to some extent on function symbols having a natural ``typing''.  That is, a given state occurs in relatively few argument positions of a function symbol.  This keeps the size of sets manipulated in the algorithm relatively small.  In addition, as noted earlier, in many naturally occurring FTAs the language associated with two states $q_1$ and $q_2$, that is, $L(q_1)$ and $L(q_2)$, tend to be either disjoint or one is included in the other.  This keeps the number of determinised states down.  Neither of these properties can be expected to hold in random FTAs.

\begin{figure}
\begin{center}
\includegraphics[width=3.5 in]{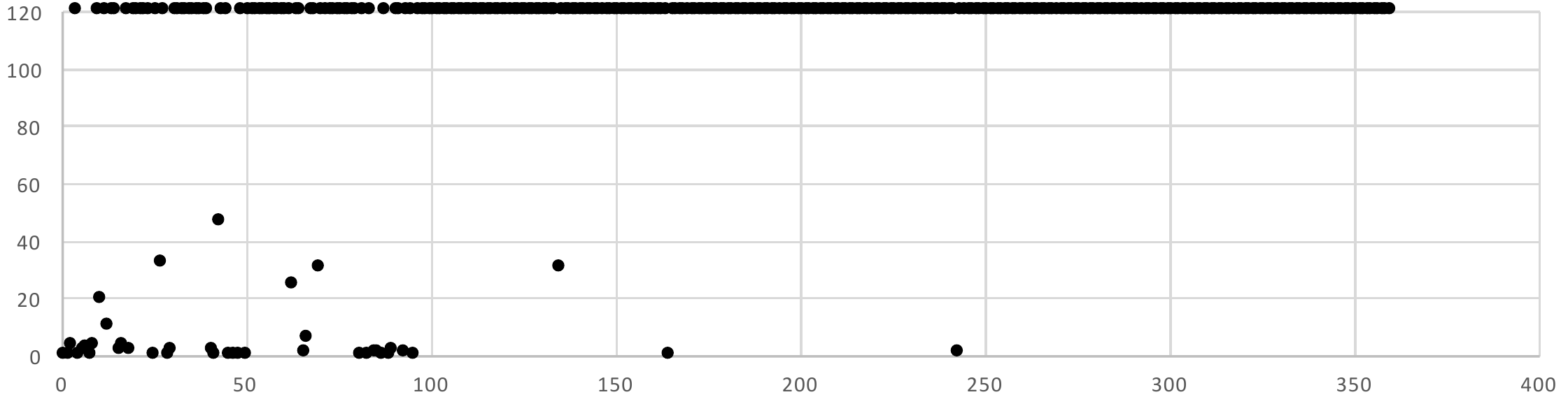}
\end{center}
\caption{Time in seconds for determinisation and completion using DFTA-opt on 360 random FTAs (timeout 120 secs.). }
\protect\label{fig-chart5}
\end{figure}

 %
\section{Applications}\label{applications}

We discuss two kinds of application: those where we need the set of transitions and those where we do not.  For the first kind, 
 applicability depends partly on whether the product form of transitions is directly usable.  The product form is of little ultimate use if the transitions need to be explicitly enumerated in order to use them.  Fortunately, it appears that the product form can often be efficiently processed directly.  
 
 \subsection{FTA operations}
 Firstly, we note that other FTA operations can be adapted to handle product form directly, rather than expanding the products. For example, the intersection of FTAs, whose transitions are in product form, is easily performed. Let  $A^1, A^2$ be FTAs $(Q_1, Q_1^{f}, \Sigma, \Delta_1)$ and $ (Q_2, Q_2^{f}, \Sigma, \Delta_2)$ respectively, where $\Delta_1$ and $\Delta_2$ are given in product form.  Then the intersection $A^1 \cap A^2$ is given by the FTA 
 $$(Q_1 \times Q_2, Q_1^{f} \times Q_2^{f}, \Sigma, \Delta_1 \times \Delta_2)$$
 \noindent
 where the transition product $\Delta_1 \times \Delta_2$ is the following set of product transitions.
 \begin{multline*}
\{f(R_1 \times S_1,\ldots,R_n \times S_n) \rightarrow (q_1,q_2) \mid \\
f(R_1,\ldots,R_n) \rightarrow q_1 \in \Delta_1, f(S_1,\ldots,S_n) \rightarrow q_2 \in \Delta_2\}\enspace.
\end{multline*}

The FTA minimisation procedure is defined on DFTAs and the operation is $O(n^2)$ where $n$ is the size of the input DFTA \cite{DBLP:conf/wia/CarrascoDF07}.  Given that our experiments show dramatic reduction of the size of DFTA when transitions are represented in product form, the adaptation of the minimisation algorithm to product form could greatly improve the scalability of the algorithm.  We have performed preliminary experiments which confirm this expectation, but a detailed study of DFTA minimisation in product form is future work.

The difference of two FTAs can also be computed and represented in product form. 
The difference of two FTAs can be computed as follows \cite{DBLP:journals/cl/KafleG17}.
Let  $A^1, A^2$ be FTAs $(Q^1, Q^1_{f}, \Sigma, \Delta^1)$ and $ (Q^2, Q^2_{f}, \Sigma, \Delta^2)$ respectively, where we assume without loss of generality that $Q^1$ and $Q^2$ are disjoint. Let $A^{1 \cup 2}$ be the union FTA $(Q^1 \cup Q^2, Q^1_{f} \cup Q^2_{f}, \Sigma, \Delta^1 \cup \Delta^2)$ and let 
 $(\Q',\Q'_f,\Sigma,\Delta')$ be the determinisation of $A^{1 \cup 2}$.  Let $\Q^2 = \{ Q' \in \Q' \mid Q' \cap  Q^2_{f} \neq \emptyset \}$.  Then $A^1 \setminus A^2 = (\Q', \Q'_f \setminus \Q^2, \Sigma,\Delta')$.  Informally, the difference is constructed by computing the automaton accepting any term in the union $L(A^1) \cup L(A^2)$ and then excluding any accepting state that can accept an element of $L(A^2)$.

 \subsection{Applications in program verification and analysis}
Gallagher \emph{et al.} \cite{Gallagher-Henriksen-ICLP04,Gallagher-Henriksen-Banda-2005} showed that program properties of interest in analysis of logic program can be formulated as sets of terms defined by an FTA on the program signature and that a precise abstract domain for static analysis could then be constructed by determinising and completing the FTA.  The algorithm presented in detail here was first developed in the context of that work.  The approach was made practical by encoding the product transitions directly as binary decision diagrams (BDDs), via a translation to Datalog form \cite{Ullman-Vol1} thus avoiding the need to enumerate transitions explicitly.

Recently the  FTA difference construction mentioned above was used to implement a refinement procedure for Horn clause verification \cite{DBLP:journals/cl/KafleG17}.  An FTA is built to represent the set of derivations in a given set of Horn clauses.  Infeasible traces can be eliminated from the FTA by application of the automata difference algorithm, which is then used to construct a new set of Horn clauses in which the infeasible trace cannot arise.

Determinisation and completion are usually needed to form the complement of an FTA, though other algorithms exist, for example implemented in the VATA library to which we compared in Section \ref{experiments}.  Applications in verification and analysis using tree automata to represent set of states could benefit from the availability of a practical complementation algorithm, e.g. \cite{BRICS-EP-00-SME_CTDC,Monniaux-SCP,GenetT01,FeuilladeGT04,BallandBEG08}.  Note that we obtain the complement with its transitions in product form, simply by switching the accepting and non-accepting states of the completed DFTA, and therefore further exploitation depends on avoiding the need to expand the product transitions.

\subsection{Applications not using the set of transitions}
For the second class of applications, that is, those for which transitions are not needed, we note that
DFTA states encode useful information in themselves, making use of the disjointness of the sets $\{L(Q_i) \mid Q_i \mathrm{~is~a ~DFTA~state}\}$ and Lemma \ref{dfta-nonempty}. For such applications, the optimised algorithm could be terminated without executing the transition generation step. We identify three applications referring only to the DFTA states.

\begin{enumerate}
\item
Firstly, an FTA $A$ over a signature $\Sigma$ is \emph{universal} if $\Lang(A) = \Term(\Sigma)$.  This is true if in the determinised, completed DFTA,  every state is an accepting state, that is, it intersects with the set of accepting states of $A$. Note that a version of the algorithm checking universality could be terminated with a negative answer as soon as a counterexample state was generated. 

\item
Secondly, the  \emph{emptiness of intersections} of input FTA states can be checked using the DFTA states.  Let $q_1,\ldots,q_n$ be FTA states. Then $L(q_1) \cap \ldots \cap L(q_n)$ is nonempty if and only if the corresponding DFTA includes a state containing $q_1,\ldots,q_n$.  

\item
Finally we consider the problem of \emph{FTA inclusion}. The classical approach to checking containment of FTA $A_1$ in FTA $A_2$ is to check the emptiness of $A_1 \cap \bar{A_2}$ which involves the complementation of $A_2$. The procedure for constructing the difference FTA given above also functions as an inclusion check. That is, to check whether $L(A_1) \subseteq L(A_2)$ construct the states of the difference automaton $A_2 \setminus A_1$, as described above. Then $L(A_1) \subseteq L(A_2)$ if and only if the set of accepting states in the DFTA of $A_2 \setminus A_1$ is empty.  As with universality checking, the algorithm can terminate with a negative answer as soon as a counterexample state is generated.  
Other containment algorithms not requiring complementation have been presented \cite{SudaH05,HosoyaVP05} and an algorithm based on antichains \cite{DBLP:conf/wia/BouajjaniHHTV08}. 
The FTA inclusion-checking algorithms  based on antichains \cite{DBLP:conf/wia/BouajjaniHHTV08} are in general more efficient, but our experiments show that the DFTA-based algorithm is comparable for cases where the answer is negative (in other words, a counterexample is found fast). However, in general the antichain algorithm is much more effective when inclusion holds.

\end{enumerate}

%
\section{Discussion and Related Work}\label{related}

The algorithm presented in this paper was sketched by Gallagher \emph{et al.} in \cite{Gallagher-Henriksen-Banda-2005}, including the concept of product transitions. Otherwise, we do not know of other attempts to design practical algorithms for determinisation.  Previous work that used tree automata as a modelling formalism commented on the impracticality of handling complementation, due to the complexity of the determinisation and completion algorithm \cite{Monniaux-SCP,Heintze-Dagstuhl}. Available libraries for tree automata manipulation seem to implement the textbook algorithm \cite{GenetT01,MONA14,LengalSV12} for determinisation.  It would be interesting to investigate whether the technique of ``guided tree automata'' in MONA could be modified to handle product transitions.

Efficient algorithms for FTA inclusion checking that avoid the need for determinisation have been developed \cite{DBLP:conf/wia/BouajjaniHHTV08}, but the operation of determinisation is still required for many other purposes.

The key aspect of the optimised algorithm is the fact that the output is generated directly in a compact form (product transitions), which can be used directly in further processing.  The space savings can be exponential, though the worst case remains the same.  The problem of FTA determinisation is inherently intractable.  As in other domains of formal reasoning, such as Boolean functions, great progress can nevertheless be achieved by finding compact representations such as BDDs \cite{DBLP:journals/tc/Bryant86} that work well on a wide range of practical cases.  Product form offers no guarantee of efficiency, but for a wide range of practical cases the compact representation makes the determinisation algorithm far more scalable.

 %
\section{Conclusion and Future Work}\label{future}

The contribution of this paper is that it is the first work to our knowledge that shows that determinisation and completion, previously considered almost hopeless cases for anything but very small FTAs, can be performed for a wide range of FTAs arising in practical applications.

There remains interesting work to do both on the algorithm itself and its applications.  Firstly, there seem to be opportunities for optimisation of the critical inner loop of the algorithm generating the DFTA states.  A state can be generated many times, and it seems likely that there are conditions on the elements of the $\Phi$ and $\Psi$ arrays in the algorithm that could be checked in order to avoid this.  The challenge is to simplify the checks sufficiently to make them worthwhile as an optimisation.  Perhaps a completely different representation of the $\Phi$ and $\Psi$ arrays, such as some Boolean encoding, is needed.  We are actively investigating this.

Secondly, we are looking at other applications of the algorithm.  The original motivating application, that of logic program analysis, is still interesting, since Horn clauses (pure logic programs) are increasingly used as a representation language for a variety of other languages and computational formalisms. Essentially the same analysis problems arise in term rewriting systems, where a system state is represented by a term, and an FTA expresses state properties of interest.

 
\section*{Acknowledgements}
We would like to thank Kim Steen Henriksen and Gourinath Banda for discussions in the early stages of this work.

\end{document}